\documentclass[11pt,letter]{article}
\usepackage[papersize={8.5in,11in},margin=1in]{geometry}
\usepackage{amsthm} 
\usepackage{times}
\usepackage{fancyhdr} 
\usepackage{color} 
\usepackage{changepage}
\usepackage[nofillcomment,linesnumbered,noend,noresetcount,noline]{algorithm2e}

\AtBeginDocument{%
 \abovedisplayskip=2pt plus 1pt minus 1pt
 \abovedisplayshortskip=0pt plus 1pt
 \belowdisplayskip=2pt plus 1pt minus 1pt
 \belowdisplayshortskip=0pt plus 1pt
}
\usepackage[noindentafter]{titlesec}
\titlespacing{\paragraph}{%
  0pt}{
  0.1\baselineskip}{
  1em}
\titlespacing\section{0pt}{8pt plus 1pt minus 1pt}{2pt plus 1pt minus 1pt}
\titlespacing\subsection{0pt}{8pt plus 1pt minus 1pt}{2pt plus 1pt minus 1pt}
\titlespacing\subsubsection{0pt}{8pt plus 1pt minus 1pt}{2pt plus 1pt minus 1pt}
\usepackage{amsmath} 
\usepackage{amssymb} 
\usepackage{graphicx} 
\usepackage{complexity} 
\usepackage{enumitem} 
\usepackage{shorttoc} 
\usepackage{array} 
\usepackage{float}
\usepackage{pdfsync}
\usepackage{verbatim}

\usepackage[pagebackref=true,colorlinks]{hyperref}
\hypersetup{linkcolor=red,filecolor=red,citecolor=red,urlcolor=red}

\newtheoremstyle{slplain}
  {.4\baselineskip\@plus.1\baselineskip\@minus.1\baselineskip}
  {.3\baselineskip\@plus.1\baselineskip\@minus.1\baselineskip}
  {\itshape}
  {}
  {\bfseries}
  {.}
  { }
  {}
\theoremstyle{slplain} 

\newtheorem*{definition*}{Definition}
\newtheorem*{theorem*}{Theorem}
\newtheorem{theorem}{Theorem}[section]
\newtheorem{lemma}[theorem]{Lemma}

\newtheorem{claim}[theorem]{Claim}
\newtheorem{corollary}[theorem]{Corollary}
\newtheorem{definition}[theorem]{Definition}

\makeatletter
\newtheorem*{rep@theorem}{\rep@title}
\newcommand{\newreptheorem}[2]{%
\newenvironment{rep#1}[1]{%
 \def\rep@title{#2 \ref{##1}}%
 \begin{rep@theorem}}%
 {\end{rep@theorem}}}
\makeatother

\newreptheorem{theorem}{Theorem}
\newreptheorem{lemma}{Lemma}
\newreptheorem{claim}{Claim}
\newreptheorem{corollary}{Corollary}

\theoremstyle{definition}

\theoremstyle{remark}

\numberwithin{equation}{section}

\newtheoremstyle{etplain}
  {.0\baselineskip\@plus.1\baselineskip\@minus.1\baselineskip}
  {.0\baselineskip\@plus.1\baselineskip\@minus.1\baselineskip}
  {\itshape}
  {}
  {\bfseries}
  {.}
  { }
  {}

\newcommand{\idlow}[1]{\mathord{\mathcode`\-="702D\it #1\mathcode`\-="2200}}
\newcommand{\id}[1]{\ensuremath{\idlow{#1}}}
\newcommand{\litlow}[1]{\mathord{\mathcode`\-="702D\sf #1\mathcode`\-="2200}}
\newcommand{\lit}[1]{\ensuremath{\litlow{#1}}}

\newcommand{\namedref}[2]{\hyperref[#2]{#1~\ref*{#2}}}

\newcommand{\sectionref}[1]{\namedref{Section}{#1}}
\newcommand{\appendixref}[1]{\namedref{Appendix}{#1}}
\newcommand{\theoremref}[1]{\namedref{Theorem}{#1}}

\newcommand{\figureref}[1]{\namedref{Figure}{#1}}
\newcommand{\figurerefb}[2]{\hyperref[#1]{Figure~\ref*{#1}#2}}

\newcommand{\lemmaref}[1]{\namedref{Lemma}{#1}}
\newcommand{\claimref}[1]{\namedref{Claim}{#1}}

\newcommand{\corollaryref}[1]{\namedref{Corollary}{#1}}

\newcommand{\equationref}[1]{\hyperref[#1]{(\ref*{#1})}}
\renewcommand{\eqref}{\equationref}

\usepackage[textsize=tiny]{todonotes}

\newcommand{\DEBUG}[1]{}

\begin{document}

\linespread{0.97}

\title{How to Elect a Leader Faster than a Tournament}
\date{\today}
\author{
Dan Alistarh \\ {\small Microsoft Research}
\and
Rati Gelashvili \\ {\small MIT}
\and
Adrian Vladu \\ {\small MIT}
}
\date{}
\maketitle
\begin{abstract}
The problem of electing a leader from among $n$ contenders is 
one of the fundamental questions in distributed computing. 
In its simplest formulation, the task is as follows: 
given $n$  processors, all participants must eventually return a \emph{win} or \emph{lose} indication, such that a single contender may \emph{win}. 
Despite a considerable amount of work on leader election, the following question is still open: can we elect a leader in an asynchronous fault-prone system faster than just running a $\Theta(\log n)$-time tournament, against a strong adaptive adversary?

In this paper, we answer this question in the affirmative, improving on a decades-old upper bound. 
We introduce two new algorithmic ideas to reduce the time complexity of electing a leader to $O( \log^* n)$, using $O( n^2 )$ point-to-point messages. 
A non-trivial application of our algorithm is a new upper bound for the \emph{tight renaming} problem, assigning $n$ items to the $n$ participants in expected $O( \log^2 n )$ time and $O(n^2)$ messages. 
We complement our results with lower bound of $\Omega( n^2 )$ messages for solving these two problems, closing the question of their message complexity.

\vspace{10em}

\end{abstract}

\thispagestyle{empty}

\newpage
\setcounter{page}{1}
\section{Introduction}
The problem of picking a leader from among a set of $n$ contenders in a fault-prone  system is among the most well-studied questions in distributed computing. 
In its simplest form,  \emph{leader election (test-and-set)}~\cite{Afek92} is stated as follows. 
Given $n$ participating processors,  each of the contenders must eventually return either a \emph{win} or \emph{lose} indication, with the property that a single participant may  \emph{win}. 
Leader election is one of a set of canonical problems, or \emph{tasks}, whose solvability and complexity are the focus of distributed computing theory, along with \emph{consensus (agreement)}~\cite{LSP82, PSL80}, \emph{mutual exclusion}~\cite{Dijkstra65}, \emph{renaming}~\cite{ABDPR90}, or \emph{task allocation (do-all)}~\cite{KS92}. These problems are usually considered in \emph{asynchronous} models, such as message-passing or shared-memory~\cite{Lynch97}. 

We focus on leader election in the \emph{asynchronous message-passing} model, 
  in which each of $n$ processors is connected to every other processor via a point-to-point channel.
Communication is asynchronous, i.e., messages can be arbitrarily delayed.
Moreover, local computation of processors is also performed in asynchronous steps.
The scheduling of computation steps and message deliveries in the system is controlled by 
  a \emph{strong (adaptive) adversary}, which can examine local state, 
  including random coin flips, and crash $t < n / 2$ of the participants at any point during the computation.
The natural complexity metrics are \emph{message complexity}, i.e., total number of messages sent by the protocol, 
  and \emph{time complexity}, i.e. the number of times a processor relies on the adversary to schedule 
  a computation step or to deliver messages.

Many fundamental results in  distributed computing are related to the complexity of canonical tasks in asynchronous models. 
For example, Fisher, Lynch, and Patterson~\cite{FLP85} showed that it is impossible to solve consensus deterministically  in an asynchronous system if one of  the $n$ participants may fail by crashing. This deterministic impossibility extends to leader election~\cite{Her91}. 
Since the publication of the FLP result, a tremendous amount of research effort has been invested into overcoming this impossibility for canonical tasks. 
 Seminal work by Ben-Or~\cite{BenOr83} showed that relaxing the problem specification to allow probabilistic termination can circumvent FLP, and obtain efficient distributed algorithms. 

Consequently, the past three decades have seen a continuous quest to improve the randomized upper and lower bounds for  canonical tasks, and in fact tight (or almost tight) complexity bounds are now known, against a strong adversary, for consensus~\cite{AC08, AAKS14}, mutual exclusion~\cite{HW09, HW10, GW12}, renaming~\cite{AACHGG14}, and task allocation~\cite{BussKaRa96, ABGG12}. 

For leader election against a strong adversary, the situation is different. 
The fastest known solution is more than two decades old~\cite{Afek92}, and is a \emph{tournament tree}: pair up the participants into  two-processor ``matches,'' decided by two-processor randomized consensus;  winners continue to compete, while losers drop out, until a single winner prevails. 
Time complexity is logarithmic, as the winner has to communicate at each tree level. 
No time lower bounds are known. 
Despite significant recent interest and progress on this problem in weaker adversarial models~\cite{AAGGG10, AlistarhA11, GW12b}, the question of whether a tournament is optimal when elect a leader against a strong adversary is surprisingly still open. 

\paragraph{Contribution.}
In this paper, we show that it is possible to break the logarithmic barrier in the classic asynchronous message-passing model, against an adaptive adversary. 
We present a new randomized algorithm 
which elects a leader in expected $O( \log^* n )$ time, sending $O( n^2 )$ messages. 

The algorithm is based on two new ideas, which we briefly describe below. 
The general structure is rather simple: computation occurs in \emph{phases}, 
where each phase is designed to drop as many participants as possible, while ensuring that at least one processor survives. 
Consider a simple implementation: 
each processor flips a biased coin at the beginning of the phase, to decide whether to give up (value $0$) or continue (value $1$), and communicates its choice to others. If at least one processor out of the $n_r$ participants in phase $r$ flips $1$, all processors which flipped $0$ can safely drop from contention. 
We could aim for $o( \log n )$ iterations by setting the probabilities to obtain less than a constant fraction of survivors in each phase.
Unfortunately, a strong adversary easily breaks such a strategy: since it  can see the flips, it can schedule all the processors that flipped $0$ to complete the phase \emph{before} any processor that flipped $1$, forcing \emph{everyone} to continue. 

\paragraph{Techniques.} Our first algorithmic idea 
is a way to \emph{hide} the processor coin flips during the phase, handicapping the adaptive adversary. 
In each phase, each processor first takes a ``poison pill" (moves to \emph{commit} state), and broadcasts this to all other processors. 
The processor then flips a biased local coin to decide whether to drop out of contention (\emph{low} priority) or to take an ``antidote" (\emph{high} priority), broadcasts its new state, and checks the states of other processors. 
Crucially, if it has flipped low priority, and sees \emph{any other processor} either in \emph{commit} state or in \emph{high priority} state, the processor returns \emph{lose}. Otherwise, it survives to the next phase.

The above mechanics guarantee at least one survivor
  (in the unlikely event where all processors flip \emph{low} priority, they all survive),  
  but can lead to few survivors in each phase.
The insight is that, to ensure many survivors, the adversary must examine the processors' coin flips. 
But to do so, the adversary must first allow it to take the poison pill (state \emph{commit}). 
Crucially, any low-priority processor observing this \emph{commit} state automatically drops out. 
We prove that, because of this catch-22, the adversarial scheduler can do no more than to let processors execute each phase sequentially, one-by-one, hoping that the first processor flipping high priority, which eliminates all later low-priority participants, comes as late as possible in the sequence. 
Now we can bias the flips such that a group of at most $O( \sqrt {n_r} )$ processors survive because they flipped high priority, 
and $O( \sqrt{ n_r} )$ processors survive because they did not observe any high priority. 
This choice of bias seems hard to improve, as it yields the perfect balance between the sizes of the two groups of survivors.

Our second algorithmic idea breaks this roadblock. 
Consider two extreme scenarios for a phase: first when all participants communicate with each other, leading to similar views 
  and second, when processors see fragmented views, observing just a subset of other processors. 
In the first case, each processor can safely set a low probability of surviving. 
This does not work in the second case since processor views have a lot of variance. We exploit this variance to break symmetry. 
Our technical argument combines these two strategies such that we obtain at most $O( \log^2 n_r)$ expected survivors in a phase, under \emph{any} scheduling. 

The final algorithm has additional useful properties.
It is \emph{adaptive}, meaning that, if $k \leq n$ processors participate, 
its complexity becomes $O( \log^* k )$. 
Moreover, since most participants drop in the first round of broadcast, the message complexity 
 is $O( k n )$, which we shall prove is asymptotically optimal.

\paragraph{Renaming.} We build on these properties to design a message-optimal algorithm for \emph{strong renaming}, 
  which assigns distinct items (or names) labeled from $1$ to $n$ to the $n$ processors, 
  using expected $O( n^2 )$ messages and $O( \log^2 n)$ time.
We employ a simple strategy: 
 each processor repeatedly picks a random name that it sees as available, announces it, 
 and competes for it via an instance of leader election. 
If the processor wins, it returns the name; otherwise, it tries again.   
The algorithm can be seen as a balls-into-bins game, in which $n$ balls are the processors and bins are the names.
We need to characterize two parameters: the maximum number of trials by a single processor, 
  and the maximum contention on a single bin, as they are linked with message and time complexity. 
The critical difficulty is that, since rounds are not synchronised, 
 the bin occupancy views perceived by the processors are effectively under adversarial control and out-of-date or incoherent views can lead 
 to wasted trials and increased contention on the bins. 
 
Our task is to prove that, in fact, this balls-into-bins process is robust to the correlations 
 and skews in the trial probability distributions caused by asynchrony.
Our approach is to carefully bound the evolution of processors' views and their trial distributions as more and more trials are performed. 
Roughly, for $j \geq 1$, we split the execution into time intervals, where at most $n / 2^{j - 1}$ names are available at the beginning of the interval $j$, and focus on bounding the number of wasted trials in each interval. 
The main technical difficulty that we overcome is that the views corresponding to these trials could be highly correlated, 
 as the adversary may delay messages to increase the probability of a collision. 

\paragraph{Lower bound.} We match the message complexity of our algorithms with an $\Omega( n^2 )$ lower bound on the expected message complexity 
 of any leader election or renaming algorithm.
The intuition behind the bound is that no processor should be able to decide without receiving a message, 
 as the others might have already elected a winner; since the adversary can fail up to $n / 2$ processors, 
 it should be able to force each processor to either send or receive $n / 2$ messages. 
However, this intuition is not entirely correct, as groups of processors could employ complex gossip-like message distribution strategies 
 to guarantee that at least \emph{some} processors receive \emph{some} messages while keeping the total message count $o(n^2)$. 
We thwart such strategies with a non-trivial indistinguishability argument, showing that in fact there must exist a group of $\Theta(n)$ processors,
 each of which either sends or receives a total of $\Theta(n)$ messages. 
A similar argument yields the $\Omega( n^2 )$ lower bound for renaming, and in fact for any object with strongly non-commutative operations~\cite{LawsOfOrder}.  

\paragraph{Related Work.} 
We focus on previous work on the complexity of randomized leader election\footnote{Some older references, e.g.~\cite{Afek92}, employ the name \emph{test-and-set} exclusively for this task, and use leader election for the consensus (agreement) problem, while more recent ones~\cite{GW12b} equate test-and-set and leader election.} and renaming, most of which considered the asynchronous shared-memory. 
However, one option is to emulate efficient shared-memory solutions via simulations between shared-memory and message-passing~\cite{ABD95}. 
This preserves time complexity, but communication may be increased by at most a linear factor. 
 
We classify previous solutions according to their adversarial model. 
Against a strong adversary, the fastest known leader election algorithm is the tournament tree of Afek et al.~\cite{Afek92}, 
 whose contention-adaptive variant was given in~\cite{AAGGG10}. 
For $n$ participants, these algorithms require $\Theta( \log n )$ time, 
  and $\Theta (n^2 \log n )$ messages using a careful simulation. 
\emph{PoisonPill} is contention-adaptive, improves time complexity (more than) exponentially, and gives tight message complexity bounds.  

For renaming, the fastest known shared-memory algorithm~\cite{AACHGG14} can be simulated with $O( \log n )$ time, 
 and $\Theta( n^2 \log n )$ messages. (The latter bounds  are  obtained by simulating an AKS sorting network~\cite{AKS}; 
 constructible solutions pay an extra logarithmic factor in both measures.)   
Our balls-into-bins approach is simpler and message-optimal, at the cost of an extra logarithmic factor in the time complexity. 
Reference~\cite{AAGGG10} uses a simpler balls-into-bins approach for renaming, where each processor tries all the names, in random order, until acquiring some one. Despite the similarity, this algorithm has expected time complexity $\Omega(n)$, as a late processor 
 may try out a linear number of spots before succeeding.

References~\cite{AlistarhA11, GW12b} considered the complexity of leader election against a weak (oblivious) adversary, 
 which fixes the schedule in advance. 
The structure of splitting the computation into sifting rounds, eliminating more than a constant factor of the participants per round, 
 was introduced in~\cite{AlistarhA11}, where the authors give an algorithm with $O(\log \log n)$ time complexity. 
Giakkoupis and Woelfel~\cite{GW12b} improved this to $O( \log^* k )$, where $k$ is the number of participants. 
These algorithms yield the same bounds in asynchronous message-passing, but their complexity bounds only hold against a \emph{weak} adversary.

The \emph{consensus} problem can be stated as leader election if we ask processors to return the \emph{identifier} of the winner, 
as opposed to a win/lose indication. As such, consensus solves leader election, but not vice-versa~\cite{Her91}. 
In fact, randomized consensus has $\Omega( n )$ time complexity~\cite{AC08}. 
Recent work~\cite{AAKS14} considered the message complexity of randomized consensus in the same model, 
achieving $O( n^2 \log^2 n )$ message complexity, and $O( n \log^3 n )$ time complexity, 
  using completely different techniques. 
\section{Definitions and Notation}
\paragraph*{System Model.}
We consider the classic asynchronous message-passing  model~\cite{ABD95}.
Here, $n$ processors communicate with each other 
        by sending \emph{messages} through \emph{channels}.
There is one channel 
        from each processor to every other processor; the channel from  $i$ to $j$ 
        is independent from the channel from $j$ to $i$.
Messages can be arbitrarily delayed by a channel, but do not get corrupted.

Computations are modeled as sequences of steps of the processors, which can be either \emph{delivery steps}, representing the delivery of a new message, or \emph{computation steps}.  
At each computation step, the processor receives all messages delivered to it since the last computation step, and, unless it is \emph{faulty}, 
  it can perform local computation and send new messages. 
  A processor is \emph{non-faulty}, if it is allowed
        to perform local computations and send messages infinitely often
        and if all messages it sends are eventually delivered.
 Notice that 
        messages are also delivered to \emph{faulty} processors, although their outgoing messages may be dropped.

We consider algorithms that tolerate up to $t \leq \lceil n/2 \rceil - 1$
        processor failures. 
That is, when more than half of the processors are non-faulty,
        they all return an answer from the protocol with probability one.
A standard assumption in this setting is that all non-faulty processors always take part in the computation
 by replying to the messages, irrespective of whether they participate in a certain algorithm 
 or even after they return a value---otherwise, the $t \leq \lceil n/2 \rceil - 1$ condition may be violated. 

\paragraph*{The Communicate Primitive.} 
Our algorithms use a procedure called \lit{communicate}, 
        defined in ~\cite{ABD95} as a building block for asynchronous communication.
The call \lit{communicate}($m$) 
        sends the message $m$ to all $n$ processors and waits for at least 
        $\lfloor n/2 \rfloor + 1$ acknowledgments before proceeding with the protocol.
The \lit{communicate} procedure can be viewed as a best-effort broadcast mechanism; 
        its key property is that any two \lit{communicate} calls intersect in at least one recipient.
In the following, a processor $i$  will \lit{communicate} messages of the form 
         (\emph{propagate},$v_i$) or (\emph{collect},$v$).
For the first message type, each recipient $j$ updates its view of the variable $v$ and  
        acknowledges by sending back an \emph{ACK} message.
In the second case, the acknowledgement is a pair (\emph{ACK},$v_j$)
        containing $j$'s view of the variable for the receiving process.
In both cases, processor $i$ waits for $> n / 2$ \emph{ACK} replies 
        before proceeding with its protocol.
In the case of \emph{collect}, the \lit{communicate} call returns an array of at least 
        $\lfloor n/2 \rfloor + 1$ views that were received.
	
\paragraph*{Adversary.} 
We consider strong adversarial setting where the scheduling of processor steps, 
        message deliveries and processor failures are controlled 
        by an adaptive adversary.
At any point, the adversary can examine the system state, 
        including the outcomes of random coin flips,
        and adjusts the scheduling accordingly.

\paragraph*{Complexity Measures.} 
We consider two worst-case complexity measures against the adaptive adversary.
\emph{Message complexity} is the maximum expected number of messages sent by all processors during an execution.
When defining \emph{time complexity}, we need to take into account the fact that, in asynchronous message-passing, 
 the adversary schedules both message delivery and local computation. 

\begin{definition*}[Time Complexity]
  Assume that the adversary fixes two arbitrarily large numbers $t_1$ and $t_2$ before an execution,
    and these numbers are unknown to the algorithm.
  Then, during the execution, the adversary delivers every message of a non-faulty processor within time $t_1$ and
    schedules a subsequent step of any non-faulty processor in time at most $t_2$.\footnote{Note that 
	the adversary can set $t_1$ or $t_2$ arbitrarily large, unknown to the algorithm, 
        so the guarantees from the algorithm's prospective are still only that messages are \emph{eventually} delivered
        and steps are \emph{eventually} scheduled.}
  An algorithm has time complexity $O(T)$ if the maximum expected time before all non-faulty processors return 
    that the adversary can achieve is $O(T (t_1 + t_2))$.\footnote{Applied to asynchronous shared-memory, this yields an alternative definition of \emph{step (time) complexity}, taking $t_2$ as an upper bound on the time for a thread to take a shared-memory step (and ignoring $t_1$). Counting all the delivery and non-trivial computation \emph{steps} in message-passing  gives an alternative definition of message complexity, corresponding to shared-memory \emph{work complexity}.}
\end{definition*}
\noindent For instance, in our algorithms, all messages  are triggered by the \lit{communicate} primitive.
A processor depends on the adversary to schedule a step in order to compute and call \lit{communicate}, 
  and then depends on the adversary to deliver these messages and acknowledgments.
In the above definition, if all processors call \lit{communicate} at most $T$ times, 
  then all non-faulty processors return in time at most $2 T (t_1 + t_2) = O(T (t_1 + t_2))$:
  each communicated message reaches destination in time $t_1$, 
  gets processed within time $t_2$, at which point the acknowledgment is sent back and delivered after $t_1$ time.
So, after $2t_1 + t_2$ time responses from more than half processors are received, 
  and in at most $t_2$ time the next step of the processor is scheduled when it again computes and communicates.
This implies the following. 
\begin{claim}
\label{clm:comtime}
For any algorithm, if the maximum expected number of \lit{communicate} calls by any processor that the adversary can achieve is $O(T)$, 
  then time complexity is also $O(T)$. 
\end{claim}
		
\paragraph*{Problem Statements.} In  \emph{leader election} (\emph{test-and-set}), each processor may return either \emph{WIN} or \emph{LOSE}. Every (correct) processor should return (\emph{termination}), and only one processor may return \emph{WIN} (\emph{unique winner}). 
No processor may lose before the eventual winner starts its execution. 
The goal is to ensure that  operations are \emph{linearizable}, i.e., can be ordered 
        such that (1) the first operation is $\id{WIN}$ and every other return value is $\id{LOSE}$, and 
        (2) the order of non-overlapping operations is respected. 
\emph{Strong (tight) renaming} requires every (correct) processor to eventually return a \emph{unique} name between $1$ and $n$.

\section{The Leader Election Algorithm}
\label{sec:tas}
Our leader election algorithm guarantees that 
        if $k$ processors participate, the maximum expected number of \lit{communicate} calls
        by any processor that the \emph{strong adaptive}  adversary can achieve is 
        $O( \log^{\ast} k )$, and the maximum expected total number of messages is $O(nk)$.
We start by illustrating the main algorithmic idea.
\subsection{The PoisonPill Technique}
Consider the protocol specified in~\figureref{fig:PoisonPill}
        from the point of view of a participating processor.
The procedure receives the id of the processor as an input, 
        and returns a $\id{SURVIVE}/\id{DIE}$ indication. 
All $n$ processors react to received messages by replying with acknowledgments 
        according to the \lit{communicate} procedure.
\begin{figure}[!h]
\DontPrintSemicolon
\hrule
{\centering
{\small
\begin{algorithm}[H]
\KwIn{Unique identifier $i$ of the participating processor}
\KwOut{$\id{SURVIVE}$ or $\id{DIE}$}
\SetKwInput{KwVariables}{Local variables}
\KwVariables{
\;$\id{Status}[n] = \{\bot\};$ \Indp
\;$\id{Views}[n][n];$
\;\textbf{int} $\id{coin};$}
\textbf{procedure} $\lit{PoisonPill}\langle i \rangle$\;
{ 
\Indp
                $\id{Status}[i] \gets \id{Commit}$\nllabel{line:commit}
                        \tcc*[r]{commit to coin flip}
        
                $\lit{communicate}(\id{propagate},\id{Status}[i])$\nllabel{line:commitpr} 
                        \tcc*[r]{propagate status}

                $coin \gets \lit{random}(1 \textit{ with probability } 1/\sqrt{n},
                                         0 \textit{ otherwise})$\nllabel{line:random}
                        \tcc*[r]{flip coin}
        
                \lIf{$\id{coin} = 0$} 
                { 
                        $\id{Status}[i] \gets \id{Low-Pri}$ 
                } \lElse 
                {
                        $\id{Status}[i] \gets \id{High-Pri}$
                }
                $\lit{communicate}(\id{propagate},\id{Status}[i])$\nllabel{line:pripr}
                        \tcc*[r]{propagate updated status}                

                $\id{Views} \gets \lit{communicate}(\id{collect},\id{Status})$ \nllabel{line:pricol}
                        \tcc*[r]{collect status from $> n/2$}

                \If{$\id{Status}[i] = \id{Low-Pri}$}
                {
                        \If{$\exists \,proc.\, j: 
                                (\exists k: \id{Views}[k][j] \in \{\id{Commit},\id{High-Pri}\}$ 
                                \textbf{and}
                                $\forall k': \id{Views}[k'][j] \neq \id{Low-Pri})$}
                        {       
                                \textbf{return} $\id{DIE}$ \nllabel{line:sift}  
                \tcc*[r]{$i$ has low priority, sees processor $j$ with either high priority
                or committed and not low priority, and dies}
                        }
                }

           \textbf{return} $\id{SURVIVE}$\;
\Indm
}     
\end{algorithm}}}
\hrule
\caption{\lit{PoisonPill} Technique}
\label{fig:PoisonPill}
\end{figure}
In the following, we call a \emph{quorum} any set of more than $n/2$ processors. 

Each participating processor announces that it is about to flip a random coin by moving to state 
        \emph{Commit} (lines~\ref{line:commit}-\ref{line:commitpr}), then obtain either low or high 
        priority based on the outcome of a biased coin flip. 
The processor then propagates its priority information to  a quorum (line~\ref{line:pripr}).
Next, it collects the status of other processors from a quorum using the $\lit{communicate}( \id{collect}, \id{Status})$ 
 call on line~\ref{line:pricol} that requests views of the array $\id{Status}$ from each processor $j$, 
 returning the set of replies received, of size at least $ n / 2$.

The crux of the round procedure is the \emph{DIE} condition on line~\ref{line:sift}. 
A processor $p$ returns $\id{DIE}$ at this line if \emph{both} of the following occur: 
        (1) the processor $p$ has low priority,
        \emph{and} (2) it observes another processor $q$ that does not have low priority        
        in any of the views, but $q$ has either high priority or is committed
        to flipping a coin (state \emph{Commit}) in some view. 
Otherwise, processor $p$ survives.
The first key observation is that
\begin{claim}
\label{clm:leastone}
If all processors participating in \lit{PoisonPill} return, at least one processor survives.
\end{claim}
\begin{proof}
Assume the contrary.
Since processors with high priority always survive, all participating processors must have a low priority.
All participants propagate their low priority information to a quorum
        by calling the \lit{communicate} procedure on line~\ref{line:pripr}.
Let $i$ be the last processor that completes this \lit{communicate} call.
At this point, the status information of all participants is already propagated to a quorum. 
More precisely, for every paritipating processor $j$,
        more than half of the processors have a view $\id{Status}[j] = \id{Low-Pri}$.

Therefore, when processor $i$ proceeds to line~\ref{line:pricol} 
        and collects the $\id{Status}$ arrays from more than half of the processors, then, 
        since any two quorums intersect, for every participating processor $j$, 
        there will be a view of some processor $k'$ showing $j$'s low priority.
All non-participating processors will have priority $\bot$ in all views.
But given the structure of the protocol, 
        processor $i$ will \emph{not} return on line~\ref{line:sift} and will survive.
This contradiction completes the proof.
\end{proof}
\noindent On the other hand, we can bound the maximum expected number of processors that survive:
\begin{claim}
\label{clm:survivor}
The maximum expected number of processors that return $\id{SURVIVE}$ is $O(\sqrt{n})$.
\end{claim}
\begin{proof}
Consider the random coin flips on line~\ref{line:random} and let us highlight the first time 
 when some processor $i$ flips value $1$.
We will argue that no other processor $j$ that subsequently (or simultaneously) flips value $0$ 
        can survive. 
Consider such a state.
When processor $j$ flips $0$, processor $i$ has already propagated its
        \id{Commit} status to a quorum of processors. 
Furthermore, processor $i$ has a high priority, 
        thus no processor can ever view it as having a low priority.
Hence, when processor $j$ collects views from a quorum, because every two quorums 
        have an intersection, some processor $k$ will definitely report the status of processor $i$
        as \id{Commit} or \id{High-Pri} and no processor will report \id{Low-Pri}.
Thus, processor $j$ will have to return $\id{DIE}$ on line~\ref{line:sift}.

The above argument implies that processors survive either if they flip $1$ and get a high priority, 
        or if they flip $0$ strictly before any other processor flips $1$.
Each of the at most $n$ processors independently flips a biased coin on line~\ref{line:random} 
        and hence, the number of processors that flip $1$ is at most the number of $1$'s 
        in $n$ Bernoulli trials with success probability $1/\sqrt{n}$, in expectation $\sqrt{n}$.
Processors that flip $0$ at the same time as the first $1$ do not survive, and it also takes 
        $\sqrt{n}$ trials in expectation before the first $1$ is flipped, giving an upper bound 
        $\sqrt{n}$ on the maximum expected number of processors that can flip $0$ and survive.
\end{proof}
\noindent It is possible to apply this technique recursively with some extra care and construct 
        an algorithm with an expected $O(\log \log{n})$ time complexity.
But we do not want to stop here.

\subsection{Heterogeneous PoisonPill}
Building a more efficient algorithm based on the \lit{PoisonPill} technique
        requires reducing the number of survivors beyond $\Omega(\sqrt{n})$ 
        without violating the invariant that not all participants may die.
We control the coin flip bias, but setting the probability of flipping $1$ to $1 / \sqrt{n}$ 
        is provably optimal.
Let the adversary schedule processors to execute \lit{PoisonPill} sequentially.
With a larger probability of flipping $1$, more than $\sqrt{n}$ processors 
        are expected to get a high priority and survive.
With a smaller probability, at least the first $\sqrt{n}$ processors 
        are expected to all have low priority and survive.
There are always $\Omega(\sqrt{n})$ survivors.

To overcome the above lower bound, after committing, we make each processor record the list $\ell$ 
        of all processors including itself, that have a non-$\bot$ status in some view collected 
        from the quorum.
Then we use the size of list $\ell$ of a processor to determine its probability bias.
Each processor also augments priority with its $\ell$ and propagates that as a status.
This way, every time a high or low priority of a processor $p$ is observed, 
        $\ell$ of processor $p$ is also known.
Finally, the survival criterion is modified: each processor first computes set $L$ as the union of 
        all processors whose non-$\bot$ statuses it ever observed itself, and of the $\ell$ lists
        it has observed in priority informations in these statuses.
If there is a processor in $L$ for which no reported view has low priority, 
        the current processor drops.

The algorithm is described in~\figureref{fig:PoisonPillAdapt}.
The particular choice of coin flip bias is influenced by factors that should become clear from 
        the analysis.
\begin{figure}[!h]
\DontPrintSemicolon
\hrule
{\centering
{\small
\begin{algorithm}[H]
\textbf{procedure} $\lit{HeterogeneousPoisonPill}\langle i \rangle$\;
{ 
\Indp
                $\id{Status}[i] \gets \{\lit{.stat} = \id{Commit}, \lit{.list} = \{ \} \}$
                        \nllabel{line:commit1}
                        \tcc*[r]{commit to coin flip}
        
                $\lit{communicate}(\id{propagate},\id{Status}[i])$\nllabel{line:commitpr1} 
                        \tcc*[r]{propagate status}
        
                $\id{Views} \gets \lit{communicate}(\id{collect},\id{Status})$\nllabel{line:colx}
                        \tcc*[r]{collect status from $> n/2$}
        
                $\ell \gets \{j \mid \exists k: \id{Views}[k][j] \neq \bot \}$\nllabel{line:countx}
                        \tcc*[r]{record participants}
                                        
                \lIf{$|\ell| = 1$}{$\id{prob} \gets 1$\tcc*[f]{set bias}}
        
                \lElse {$\id{prob} \gets \frac{\log{|\ell|}}{|\ell|}$\tcc*[f]{set bias}}
        
                $coin \gets \lit{random}(1 \textit{ with probability } \id{prob},
                                         0 \textit{ otherwise})$ \nllabel{line:random1}
                                \tcc*[r]{flip coin}
        
                \lIf{$\id{coin} = 0$} 
                {
                        $\id{Status}[i] \gets \{\lit{.stat}=\id{Low-Pri}, \lit{.list}=\ell\}$ 
                                \tcc*[f]{record priority and list}
                }  
        
                \lElse 
                {
                        $\id{Status}[i] \gets \{\lit{.stat}=\id{High-Pri}, \lit{.list}=\ell\}$
                                \tcc*[f]{record priority and list}
                }
        
                $\lit{communicate}(\id{propagate},\id{Status}[i])$\nllabel{line:pripr1}
                        \tcc*[r]{propagate priority and list}                

                $\id{Views} \gets \lit{communicate}(\id{collect},\id{Status})$ \nllabel{line:pricol1}
                        \tcc*[r]{collect status from $> n/2$}
        
                \If{$\id{Status}[i].\lit{stat} = \id{Low-Pri}$}
                {
                        $L \gets \cup_{k,j: \id{Views}[k][j] \neq \bot} \id{Views}[k][j].\lit{list}$
                               \nllabel{line:upd1}
                               \tcc*[r]{union all observed lists}
        
                        $L \gets L \cup \{j \mid \exists k: \id{Views}[k][j] \neq \bot \}$
                               \nllabel{line:upd2}
                               \tcc*[r]{record new participants}

                        \If{$\exists \,proc.\, j \in L: 
                                \forall k: \id{Views}[k][j].\lit{stat} \neq \id{Low-Pri}$}
                        {       
                                \textbf{return} $\id{DIE}$ \nllabel{line:sift1}  
                        \tcc*[r]{$i$ has low priority, learns about processor $j$ participating 
                        whose low priority is not reported, and dies}
                        }
                }

           \textbf{return} $\id{SURVIVE}$\;

\Indm
}     
\end{algorithm}}}
\hrule
\caption{Heterogeneous \lit{PoisonPill}}
\label{fig:PoisonPillAdapt}
\end{figure}
Despite modifications, the same argument as in~\claimref{clm:leastone} still guarantees 
        at least one survivor.
Let us now prove that the views of the processors have the following interesting \emph{closure property},
 which will be critical to bounding the number of survivors with low priority.
\begin{claim}
\label{clm:closure}
Consider any set $S$ of processors that each flip $0$ and survive.
Let $U$ be the union of all $L$ lists of processors in $S$.
Then, for $p \in U$ and every processor $q$ in the $\ell$ list of $p$, $q$ is also in $U$.
\end{claim}
\begin{proof}
In order for processors in $S$ to survive, they should have observed a low priority for each
 of the processors in their $L$ lists.
Thus, every processor $p \in U$ must flip $0$, as otherwise it would not have a low priority.
However, the low priority of $p$ observed by a survivor was augmented by the $\ell$ list of $p$. 
According to the algorithm, the survivor includes in its own $L$ all processors $q$ from this $\ell$ list of $p$,
 implying $q \in U$.
\end{proof}
\noindent Next, let us prove a few other useful claims:
\begin{claim}
\label{clm:pqorder}
If processor $q$ completed executing line~\ref{line:commitpr1} no later than processor $p$
        completed executing line~\ref{line:commitpr1}, then $q$ will be included
        in the $\ell$ list of $p$.
\end{claim}
\begin{proof}
When $p$ collects statuses on line~\ref{line:colx} from a quorum, $q$ is already done propagating 
        its \id{Commit} on line~\ref{line:commitpr1}.
As every two quorum has an intersection, $p$ will observe a non-$\bot$ status of $q$ 
        on line~\ref{line:countx}.
\end{proof}
\begin{claim}
\label{clm:z1z}
The probability of at least $z$ processors flipping $0$ and surviving is $O(1/z)$.
\end{claim}
\begin{proof}
Let $S$ be the set of the $z$ processors that flip $0$ and survive and let us define $U$ as in~\claimref{clm:closure}.
For any processor $p \in U$ and any processor $q$ that completes executing line~\ref{line:commitpr1} no later than $p$,
 by~\claimref{clm:pqorder} processor $q$ has to be contained in the $\ell$ list of $p$, which by the closure property (\claimref{clm:closure})
 implies $q \in U$.
Thus, if we consider the ordering of processors according to the time they complete executing line~\ref{line:commitpr1},
 all processors not in $U$ must be ordered strictly after all processors in $U$.

Therefore, during the execution, first $|U|$ processors that complete line~\ref{line:commitpr1}
        must all flip $0$.
The adversary may influence the composition of $U$, but by the closure property, each $\ell$ list of processors in $U$ 
        contains only processors in $U$, meaning $|\ell| \leq |U|$.
So the probability for each processor to flip $0$ is at most $(1-\frac{\log{|U|}}{|U|})$ and for
        all processors in $U$ to flip $0$'s is at most $(1-\frac{\log{|U|}}{|U|})^{|U|} = O(1/|U|)$. 
This is $O(1/z)$ since all $z$ survivors from $S$ are included in their own lists and hence also in $U$.
\end{proof}
\noindent We have never relied on knowing $n$. 
If $k \leq n $ processors participate in the heterogeneous $\lit{PoisonPill}$, we get
\begin{lemma}
\label{lem:loglow}
The maximum expected number of processors that flip $0$ and survive is $O(\log{k}) + O(1)$.
\end{lemma}

\begin{lemma}
\label{lem:loghigh}
The maximum expected number of processors that flip $1$ is $O(\log^2{k}) + O(1)$.
\end{lemma}
\begin{proof}
Consider the ordering of processors according to the time they complete executing
        line~\ref{line:commitpr1}, breaking ties arbitrarily.
Due to~\claimref{clm:pqorder}, the processor that is ordered first always has $|l| \geq 1$, 
        the second processor always computes $|l| \geq 2$, and so on.
The probability of flipping $1$ decreases as $|l|$ increases, and the best expectation achievable
        by adversary is $1 + \sum_{l=2}^{k} \frac{\log{l}}{l} = O(\log^2{k}) + O(1)$ as desired.
\end{proof}
\subsection{Final construction}
The idea of implementing leader election is to have rounds of heterogeneous \lit{PoisonPill},
        where all processors participate in the first round and only the survivors of round $r$ 
        participate in round $r+1$.
Each processor $p$, before participating in round $r_p$, first propagates $r_p$ as its current round 
        number to a quorum, then collects information about the rounds of other processors from a quorum.
Let $R$ be the maximum round number of a processor in all views that $p$ collected.
To determine the winner, we use the idea from~\cite{SSW91}: 
        if $R > r_p$, then $p$ loses and if $R < r_p - 1$ then $p$ wins.    
We also use a standard doorway mechanism~\cite{Afek92} to ensure linearizability.
The pseudocode of the final construction is given in~\appendixref{app:tasproofs}, 
        along with the complete proof of the following statement:
\begin{reptheorem}{thm:maintas}
Our leader election algorithm is linearizable.
If there are at most $\lceil n/2 \rceil -1$ processor faults, all non-faulty processors 
        terminate with probability $1$.
For $k$ participants, it has time complexity $O(\log^{\ast}{k})$ and message complexity $O(kn)$.
\end{reptheorem}
\noindent The performance guarantees follow from~\lemmaref{lem:loglow} and~\lemmaref{lem:loghigh} 
        with some careful analysis.
In particular, later rounds in which maximum expected number of participants is constant require
        special treatment.

\section{The Renaming Algorithm}
\begin{figure}[th]
\hrule
\DontPrintSemicolon
{\centering
{\small
\begin{algorithm}[H]
\KwIn{Unique identifier $i$ from a large namespace}
\KwOut{\textbf{int} $\id{name} \in [n]$}
\SetKwInput{KwVariables}{Local variables}
\KwVariables{
\;\textbf{bool} $\id{Contended}[n] = \{\lit{false}\};$ \Indp
\;\textbf{bool} $\id{Views}[n][n];$
\;\textbf{int} $\id{coin}, \id{spot}, \id{outcome};$}
\textbf{procedure} $\lit{getName}\langle i \rangle$\;
{ 
\Indp
        \While{$\id{true}$\nllabel{line:while}}
        {
                $\id{Views} \gets \lit{communicate}(\id{collect},\id{Contended})$
                        \nllabel{line:colContended}
                        \tcc*[r]{collect contention information}
        
                \For{$j \gets 1$ \textbf{to} $n$}
                {
                        \If{$\exists k: \id{Views}[k][j] = \lit{true}$}
                        {
                                $\id{Contended}[j] \gets \lit{true}$
                                       \nllabel{line:setcont}
                                       \tcc*[r]{mark names that became contended}
                        }
                }

                $\lit{communicate}(\id{propagate},\{\id{Contended}[j] 
                        \mid \id{Contended}[j] = \lit{true}\})$
                        \nllabel{line:prepropcontend}
                        \tcc*[r]{propagate}
        
                $\id{spot} \gets \lit{random}(j \mid \id{Contended}[j] = \lit{false})$
                        \nllabel{line:pickspot}
                        \tcc*[r]{pick random uncontended name}                
        
                $\id{Contended}[\id{spot}] \gets \lit{true}$
        
                $\id{outcome} \gets \lit{LeaderElect}_{\id{spot}}(i)$
                         \tcc*[r]{contend for a new name}
        
                $\lit{communicate}(\id{propagate},\id{Contended}[\id{spot}])$
                        \nllabel{line:propcontend}
                        \tcc*[r]{propagate contention}

                \If{$\id{outcome} = \id{WIN}$}
                {
                        \textbf{return} $spot$\nllabel{line:namefound}
                                \tcc*[r]{win iff you are leader}
                }
        }
\Indm
}
\end{algorithm}}}
\hrule
\caption{Pseudocode of the renaming algorithm for $n$ processors.}
\label{fig:rename}
\end{figure}
\noindent The algorithm is described in~\figureref{fig:rename}. 
There is a separate leader election protocol for each name;
        which a processor must win in order to claim the name.
Each processor repeatedly chooses a new name and contends for it by
        participating in the corresponding leader election,
        until it eventually wins.
Processors keep track of contended names and use this information
        to choose the next name to compete for:
        in particular, the next name is selected uniformly at random from the uncontended names.
The algorithm is correct.
\begin{replemma}{lem:renameterminate}
No two processors return the same name from the \lit{getName} call and if there are at most $\lceil n/2 \rceil -1$ processor faults, 
 all non-faulty processors terminate with probability 1.
\end{replemma}
\noindent Omitted proofs can be found in~\appendixref{app:renproofs}.
Let us now introduce some notation. 
For an arbitrary execution, and for every name $u$, consider the first time when more than 
        half of processors have $\id{Contended}[u] = \lit{true}$ in their view
        (or time $\infty$, if this never happens).
Let $\prec$ denote the name ordering based on these times, and 
        let $\{u_i\}$ be the sequence of names sorted according to increasing $\prec$.
Among the names with time $\infty$, sort later the ones that are never contended by the processors.
Resolve all the remaining ties according to the order of the names.
This ordering has the following useful temporal property.
\begin{replemma}{lem:renameorder}
In any execution, if a processor views $\id{Contended}[i] = \lit{true}$ 
        in some while loop iteration, 
        and in some subsequent iteration on line~\ref{line:pickspot} the same 
        processor views $\id{Contended}[j] = \lit{false}$, $i \prec j$ has to hold.
\end{replemma}
\noindent 
Let $X_i$ be a random variable, denoting the number of processors 
        that ever contend in a leader election for the name $u_i$. 
The following holds.
\begin{replemma}{lem:msgrenameweak}
The message complexity of our renaming algorithm is
        $O(n \cdot \mathbb{E}[\sum_{i=1}^n X_i])$. 
\end{replemma}
\noindent 
We partition names $\{u_i\}$ into $\log{n}$ \emph{groups}, 
        where the first group $G_1$ contains the first $n/2$ names, 
        second group $G_2$ contains the next $n/4$ names, etc.
We use notation $G_{j' \geq j}$, $G_{j'' > j}$ and $G_{j''' < j}$ to denote 
        the union of all groups $G_{j'}$ where $j' \geq j$, 
        all groups $G_{j''}$ where $j'' > j$, and all groups $G_{j'''}$ where $j''' < j$, respectively.
We can now split any execution into at most $\log{n}$ phases.
The first phase starts when the execution starts and ends as soon as for each $u_i \in G_1$ 
        more than half of the processors view $\id{Contended}[u_i] = \lit{true}$ 
        (the way $\{u\}$ is sorted, this is the same as when 
        the contention information about $u_{n/2}$ is propagated to a quorum). 
At this time, the second phase starts and ends when for each $u_i \in G_2$
        more than half of the processors view $\id{Contended}[u_i] = \lit{true}$.
When the second phase ends, the third phase starts, and so on.

Consider any loop iteration of some processor $p$ in some execution.
We say that an iteration \emph{starts} at a time instant when $p$ executes line~\ref{line:while} 
        and reaches line~\ref{line:colContended}.
Let $V_p$ be $p$'s view of the $\id{Contended}$ array
        right before picking a spot on line~\ref{line:pickspot} in the given iteration.
We say that an iteration is $\id{clean}(j)$,
        if the iteration starts during phase $j$
        and no name from later groups $G_{j'' > j}$ is contended in $V_p$.       
We say that an iteration is $\id{dirty}(j)$,
        if the iteration starts during phase $j$
        and some name from a later group $G_{j'' > j}$ is contended in $V_p$.

Observe that any iteration that starts in phase $j$ can be uniquely classified as 
        $\id{clean}(j)$ or $\id{dirty}(j)$ and in these iterations,
        processors view all names $u_i \in G_{j''' < j}$ from previous groups as contended.
\begin{replemma}{lem:ncontend}
In any execution, at most $\frac{n}{2^{j-1}}$ processors ever contend for names 
        from groups $G_{j' \geq j}$.
\end{replemma}
\begin{lemma}
\label{lem:phaseunion}
For any fixed $j$, the total number of $\id{clean}(j)$ iterations
        is larger than or equal to $\alpha n + \frac{n}{2^{j-1}}$ with probability 
        at most $e^{-\frac{\alpha n}{16}}$ for all $\alpha \geq \frac{1}{2^{j-5}}$.
\end{lemma}
\begin{proof}
Fix some $j$.
Consider a time $t$ when the first $\id{dirty}(j)$ iteration is completed.
At time $t$, there exists an $i$ such that $u_i \in G_{j''>j}$ and a quorum of processors view $\id{Contended}[i] = \lit{true}$, 
 so all iterations that start later will set $\id{Contended}[i] \gets \lit{true}$ on line~\ref{line:setcont}.
Therefore, any iteration that starts after $t$ must observe $\id{Contended}[i] = \lit{true}$ 
        on line~\ref{line:pickspot} and by definition cannot be $\id{clean}(j)$.
By~\lemmaref{lem:ncontend}, at most $\frac{n}{2^{j-1}}$ processors
        can have active $\id{clean}(j)$ iterations at time $t$.
The total number of $\id{clean}(j)$ iterations is thus upper bounded by 
        $\frac{n}{2^{j-1}}$ plus the number of $\id{clean}(j)$ 
        iterations completed before time $t$, which we denote as \emph{safe} iterations.\footnote{
If no $\id{dirty(j)}$ iteration ever completes, then we call all $\id{clean(j)}$ iterations safe.
}
Safe iterations all finish before any iteration where a processor contends for a name in $G_{j''>j}$ is completed.
By~\lemmaref{lem:ncontend}, at most $\frac{n}{2^{j}}$ different processors can ever contend 
        for names in $G_{j''>j}$, therefore, $\alpha n$ safe iterations can occur only if in 
        at most $\frac{n}{2^{j}}$ of them processors choose to contend in $G_{j''>j}$.
Otherwise, some processor would have to complete an iteration where it contended 
        for a name in $G_{j''>j}$.
        
In every $\id{clean}(j)$ iteration, on line~\ref{line:pickspot}, 
        any processor $p$ contends for a name in $G_{j' \geq j}$ uniformly at random 
        among non-contended spots in its view $V_p$.
With probability at least $\frac{1}{2}$, $p$ contends for a name from $G_{j''>j}$,
        because by definition of $\id{clean}(j)$, 
        all spots in $G_{j''>j}$ are non-contended in $V_p$.

Let us describe the process by considering a random variable 
        $Z \sim \mathrm{B}(\alpha n, \frac{1}{2})$ for $\alpha \geq \frac{1}{2^{j-5}}$,
        where each success event corresponds to an iteration contending in $G_{j''>j}$.
By the Chernoff Bound, the probability of $\alpha n$ iterations with at most
        $\frac{n}{2^{j}}$ processors contending in $G_{j''>j}$ is:
\begin{align*}
\Pr\left[Z \leq \frac{n}{2^j}\right] = 
\Pr\left[Z \leq \frac{\alpha n}{2} \left( 1 - \frac{2^{j-1} \alpha - 1}{2^{j-1} \alpha}\right)\right]
\leq \exp\left(- \frac{\alpha n (2^{j-1} \alpha - 1)^2}{2 (2^{j-1} \alpha)^2} \right)
\leq e^{-\frac{\alpha n}{8}}
\end{align*}
So far, we have assumed that the set of names belonging to the later groups $G_{j''>j}$ was fixed,
        but the the adversary controls the execution.
Luckily, what happened before phase $j$ (i.e. the actual names that were acquired from $G_{j'''<j}$)
        is irrelevant, because all the names from the earlier phases
        are viewed as contended by all iterations that start in phases $j' \geq j$.
Unfortunately, however, the adversary also influences what names belong to
        group $j$ and to groups $G_{j''>j}$.
There are $\binom{2^{1-j}n}{2^{-j}n}$ different possible choices for names in $G_j$, 
        and by a union bound, the probability that $\alpha n$ iterations can occur 
        even for one of them is at most:
\begin{align*}
e^{-\frac{\alpha n}{8}} \cdot \binom{2^{1-j}n}{2^{-j}n}
\leq e^{-\frac{\alpha n}{8}} \cdot (2e)^{2^{-j}n}
\leq e^{-n (2^{-3}\alpha - 2^{1-j})}
\leq e^{-\frac{\alpha n}{16}}, \textnormal{ proving the claim.}
\end{align*}
\end{proof}
\noindent Plugging $\alpha = \beta - \frac{1}{2^{j-1}} \geq \frac{1}{2^{j-5}}$ in the above lemma,
        we obtain that the total number of $\id{clean}(j)$ iterations 
        is larger than or equal to $\beta n$ with probability 
        at most $e^{-\frac{\beta n}{32}}$ for all $\beta \geq \frac{1}{2^{j-6}}$.
 Let $X_i(\id{clean})$ be the number of processors that ever contend for 
        $u_i \in G_j$ in some $\id{clean}(j)$ iteration and define $X_i(\id{dirty})$ analogously:
        as the number of processors that ever contend for 
        $u_i \in G_j$ in some $\id{dirty}(j)$ iteration.
Relying on the above bound on the number of $\id{clean}(j)$ iterations, we get the following result:
\begin{replemma}{lem:rentype1}
$\mathbb{E}[\sum_{i=1}^n X_i(\id{clean})] = O(n).$
\end{replemma}
\noindent Each iteration where a processor contends for a name $u_i \in G_j$ is by definition
        either as $\id{clean}(j)$, $\id{dirty}(j)$ or starts in a phase $j''' < j$.
Let us call these $\id{cross}(j)$  and denote by $X_i(\id{cross})$ the number of such iterations.
We show that in any execution, for each $j$, any processor participates in at most one $\id{dirty}(j)$
        and at most one $\id{cross}(j)$ iteration.
This allows us to prove with some work that $\mathbb{E}[\sum_{i=1}^n X_i(\id{dirty})] = O(n)$ and 
        $\mathbb{E}[\sum_{i=1}^n X_i(\id{cross})] = O(n)$ (Lemma~\ref{lem:rentype23}). 
        The message complexity upper bound then follows by piecing together the previous claims. 
\begin{theorem}
The expected message complexity of our renaming algorithm is $O(n^2)$.
\end{theorem}
\begin{proof}
We know $X_i = X_i(\id{clean}) + X_i(\id{dirty}) + X_i(\id{cross})$ 
        and by~\lemmaref{lem:rentype1},~\lemmaref{lem:rentype23}
        we get $\mathbb{E}[\sum_i X_i] = O(n)$.
Combining with~\lemmaref{lem:msgrenameweak} gives the desired result.
\end{proof}
\noindent The time complexity upper bound exploits a trade-off between the probability 
 that a processor collides in an iteration (and must continue) and the ratio of 
 available slots which must be assigned during that iteration. 
\begin{reptheorem}{thm:rounds}
The time complexity of the the renaming algorithm is $O( \log^2 n )$.
\end{reptheorem}
\section{Message Complexity Lower Bounds}
\label{sec:lb}

We can prove that the algorithms we presented for leader election and renaming in the previous two 
sections are asymptotically message-optimal. Due to space constraints, the proof of this result 
is deferred to the Appendix. 
(In fact, we show a stronger claim, proving the same message complexity lower bound for arbitrary objects with strongly non-commutative operations as defined in~\cite{LawsOfOrder}.)

\begin{repcorollary}{cor:lowerbound}
	Any implementation of leader election or renaming by $k \leq n$ 
        processors guaranteeing termination with probability at least $\alpha > 0$ in an asynchronous message-passing system where $t < n / 2$ processors may 
        fail by crashing must have worst-case expected message complexity $\Omega( \alpha k n )$. 
\end{repcorollary}

\section{Discussion and Future Work}
\label{sec:conclusion}

We have given the first sub-logarithmic leader election algorithm against a strong adversary, 
and asymptotically tight bounds for the message complexity of renaming and leader election. 
Our results also limit the power of topological lower bound techniques, e.g.~\cite{HS99}, 
when applied to randomized leader election, since these techniques allow processors to communicate using unit-cost broadcasts or snapshots. 
Our algorithm shows that no bound stronger than $\Omega( \log^* n )$ time is possible using such techniques, unless the cost of information dissemination is explicitly taken into account. 

Determining the tight time complexity bounds for leader election remains an intriguing open question. 
Another interesting research direction would be to apply the tools we developed to obtain time- and message-efficient implementations of other fundamental distributed tasks, 
 such as task allocation or mutual exclusion, and to explore solutions optimizing bit complexity.


\bibliographystyle{alpha}
\footnotesize
\bibliography{biblio}
\normalsize
\appendix
\section{Deferred Proofs}
\subsection{Leader Election Construction and Analysis}
\label{app:tasproofs}
\figureref{fig:rounds} contains the pseudocode of \lit{PreRound} procedure that processors 
        execute before participating in round $r$.
Every processor starts in the same initial non-negative round.
The \lit{PreRound} procedure takes round number $r$ and the id of the processor as an input and 
        outputs either $\id{PROCEED}$, $\id{WIN}$ or $\id{LOSE}$.
Each processor $p$ first propagates $r$ to a quorum, then collects information about the rounds of 
        other processors also from a quorum.
Let $R$ be the maximum round number of a processor in all views that $p$ collected.
Using idea from~\cite{SSW91}, if $R > r$, then $p$ loses, if $R < r - 1$ then $p$ wins and 
        otherwise $p$ returns $\id{PROCEED}$.
\begin{figure}[!h]
\DontPrintSemicolon
\hrule
{\centering
{\small
\begin{algorithm}[H]
\KwIn{Unique identifier $i$ of the participating processor, round number $r$}
\KwOut{$\id{PROCEED}$, $\id{WIN}$, or $\id{LOSE}$}
\SetKwInput{KwVariables}{Local variables}
\KwVariables{
\;\textbf{int} $\id{Round}[n] = \{0\};$ \Indp
\;\textbf{int} $\id{Views}[n][n];$
\;\textbf{int} $\id{R};$}
\BlankLine
\textbf{procedure} $\lit{PreRound}\langle i, r \rangle$\;
{ 
\Indp
        $\id{Round}[i] \gets \id{r}$
                \tcc*[r]{record own round}
        
        $\lit{communicate}(\id{propagate},\id{Round}[i])$ \nllabel{line:rprop}
                \tcc*[r]{propagate own round}

        $\id{Views} \gets \lit{communicate}(\id{collect},\id{Round})$ \nllabel{line:rcollect}
                \tcc*[r]{collect round from $> n/2$}
        
        $R \gets \max_{k,j \mid j \neq i}(Views[k][j])$
                \tcc*[r]{maximum round of other processors observed}
        
        \If{$r < R$}
        {
                \textbf{return} \nllabel{line:highround} $\id{LOSE}$\;
        }
        \If{$R < r - 1$}
        {
                \textbf{return} $\id{WIN}$\;
        }

        \textbf{return} $\id{PROCEED}$\;
\Indm
}
\end{algorithm}}}
\hrule
\caption{\lit{PreRound} procedure}
\label{fig:rounds}
\end{figure}

To ensure linearizability we use a standard doorway technique, described in~\figureref{fig:doorway}.
This doorway mechanism is implemented by the variable $\id{door}$ stored by the processors. 
A value $\lit{false}$ corresponds to the door being open and 
        a value $\lit{true}$ corresponds to the door being closed.
Each participating processor $p$ starts by collecting the views of $\id{door}$ from more than half 
        of the processors on line~\ref{line:doorcheck}.
If a closed door is reported, $p$ is too late and automatically returns $\id{LOSE}$.
The door is closed by processors on line~\ref{line:doorclose}, 
        and this information is then propagated to a quorum. 
The goal of the doorway is to ensure that no processor can lose \emph{before} the winner 
        has started its execution. 
\begin{figure}[!h]
\DontPrintSemicolon
\hrule
{\centering
{\small
\begin{algorithm}[H]
\KwOut{$\id{PROCEED}$ or $\id{LOSE}$}
\SetKwInput{KwVariables}{Local variables}
\KwVariables{
\;\textbf{bool} $\id{door} = \lit{false}$\Indp \tcc*[r]{door is initially open} 

\textbf{bool} $\id{Doors}[n];$}
\BlankLine
\textbf{procedure} $\lit{Doorway} \langle \rangle$\;
{ 
\Indp
        $\id{Doors} \gets \lit{communicate}(\id{collect},\id{door})$ \nllabel{line:doorcheck} 
                \tcc*[r]{collect $\id{door}$ from  $ > n / 2$}

        \If{$\exists j: Doors[j] = \lit{true}$}
        {
                \textbf{return} $\id{LOSE}$\nllabel{line:doorlose}      
                       \tcc*[r]{lose if the door is closed}
        }

        $\id{door} \gets \lit{true}$\nllabel{line:doorclose}
                \tcc*[f]{close the door}
        
        $\lit{communicate}(\id{propagate},\id{door})$ \nllabel{line:doorclosepr} 
                \tcc*[r]{propagates  $\id{door} = \lit{true}$ to $> n/2$}      
        
        \textbf{return} $\id{PROCEED}$\;
\Indm
}
\end{algorithm}}}
\hrule
\caption{\lit{Doorway} procedure}
\label{fig:doorway}
\end{figure}

Finally we put the pieces together.
Our complete leader election algorithm is described in~\figureref{fig:elect}.
It involves going through the doorway procedure in the beginning, and then
        rounds of $\lit{PreRound}$ procedure possibly followed by participation in 
        a $\lit{HeterogeneousPoisonPill}$ protocol for round $r$.
Note that $\lit{HeterogeneousPoisonPill}$ protocols for different rounds are completely disjoint
        from each other.
\begin{figure}[!h]
\DontPrintSemicolon
\hrule
{\centering
{\small
\begin{algorithm}[H]
\KwIn{Unique identifier $i$ of the participating processor}
\KwOut{$\id{WIN}$ or $\id{LOSE}$}
\SetKwInput{KwVariables}{Local variables}
\KwVariables{
\;\textbf{int} $\id{r = 1};$ \Indp
\;$\id{outcome}$;}
\BlankLine
\textbf{procedure} $\lit{LeaderElect}\langle i \rangle$\;
{ 
\Indp
        \If{$\lit{Doorway} \langle \rangle = \id{LOSE}$}
        {
                \textbf{return} $\id{LOSE}$\nllabel{line:eldoor}
                       \tcc*[r]{lose if door was closed}
        }
       
        \Repeat{never}{\nllabel{line:elloop}
                $\id{outcome} \gets \lit{PreRound}\langle i, r\rangle$
                       \tcc*[r]{preround routine}
        
                \If{$\id{outcome} \in \{\id{WIN}, \id{LOSE}\}$}
                {
                \textbf{return} $\id{outcome}$\nllabel{line:elrlose}
                       \tcc*[r]{return if rounds permit}
                }
        
                \BlankLine
        
                \If{$\lit{HeterogeneousPoisonPill_r\langle i \rangle} = \id{DIE}$}
                {
                \textbf{return} $\id{LOSE}$ \nllabel{line:elplose}
                       \tcc*[r]{lose if did not survive the round}
                }
        
                $r \gets r + 1$
        }
\Indm
}
\end{algorithm}}}
\hrule
\caption{Leader election algorithm}
\label{fig:elect}
\end{figure}
        
We now prove the properties of the algorithm.
\begin{lemma}
\label{lem:leastone}
If all processors that call \lit{LeaderElect} return,
        at least one processor returns $\id{WIN}$.
\end{lemma}
\begin{proof}
Assume for contradiction that all processors that participate in \lit{PoisonPill} return $\id{LOSE}$.
Let us first prove that at least one processor always reaches the loop on line~\ref{line:elloop},
        or alternatively that not all processors can lose on line~\ref{line:eldoor}.
This would mean that all processors return $\id{LOSE}$ on line~\ref{line:doorlose} of the 
        $\lit{Doorway}$ procedure, but in that case the door would never be closed 
        on line~\ref{line:doorclose}.
Thus, all processor views would be $\id{door}=\lit{false}$, and no processor would 
        actually be able to return on line~\ref{line:doorlose}.

Since we showed that at least one processor reaches the loop, let us consider the largest round $r$ 
        in which some processors return, either in the pre-round routine of round $r$ 
        on line~\ref{line:elrlose} or because of the poison pill on line~\ref{line:elplose}.
By our assumption all these processors return $\id{LOSE}$ in round $r$.
But then, none of them may return on line~\ref{line:elrlose}, because this is only possible after
        returning $\id{LOSE}$ on line~\ref{line:highround}, which only happens if a larger 
        round than $r$ is reported, contradicting our assumption that $r$ is the largest round.

So, at least one processor participates in the $\lit{HeterogeneousPoisonPill}_r$ protocol.
However, by exactly the same argument as in~\claimref{clm:leastone}, 
        $\lit{HeterogeneousPoisonPill}_r$ is guaranteed to have at least one survivor which would then 
        participate in round $r+1$, again contradicting that $r$ is the largest round.
\end{proof}
\begin{lemma}
\label{lem:mostone}
At most one processor that executes $\lit{LeaderElect}$ can return $\id{WIN}$.
\end{lemma}
\begin{proof}
A processor $p$ can only return $\id{WIN}$ from $\lit{LeaderElect}$ on line~\ref{line:elrlose}, which
        only happens after $p$ returns $\id{WIN}$ from $\lit{PreRound}$ call with some round $r$.
This means $p$ first propagated round $r$ to a quorum on line~\ref{line:rprop}, then collected
        views of $\id{Round}$ array on line~\ref{line:rcollect}, and observed maximum round 
        $R < r - 1$ of any processor in any of the views.
This implies that when $p$ finished propagating $r$ to a quorum, no processor had finished 
        propagating $r-1$, i.e. executing line~\ref{line:rprop} in round $r-1$.
Otherwise, since every two quorums have an intersection, $p$ would have observed round $r-1$
        and $R < r - 1$ would not hold.
But for every other processor $q$, when $q$ executes line~\ref{line:rcollect} in round $r-1$ and
        invokes the $\lit{PreRound}$ procedure, $R$ will be at least $r$ since $p$ has already 
        propagated to a quorum, so $q$ will observe $r-1 < r$ and return $\id{LOSE}$ 
        on line~\ref{line:highround} and subsequently return $\id{LOSE}$ from $\lit{LeaderElect}$.
\end{proof}
\begin{lemma}
\label{lem:elinear}
Our leader election algorithm is linearizable.
\end{lemma}
\begin{proof}
All processors that execute $\lit{LeaderElect}$ cannot return $\id{LOSE}$ by~\lemmaref{lem:leastone}.
Therefore, in every execution we can find $\lit{LeaderElect}$ invocation where processor either
        does not return, or returns $\id{WIN}$.
On the other hand, by~\lemmaref{lem:mostone}, no more than one processor can return $\id{WIN}$.
If no processor returns $\id{WIN}$, let us linearize the processor that invoked $\lit{LeaderElect}$
        the earliest as the leader.
This way, we always have an unique processor to be linearized as the winner.
We linearize it at the beginning of its invocation interval, say point $P$, and claim that every 
        remaining $\lit{LeaderElect}$ call can be linearized as returning $\id{LOSE}$ after $P$.

Assume contrary, then the problematic $\lit{LeaderElect}$ invocation must return before $P$, and 
        we know it has to return $\id{LOSE}$.   
By definition, this earlier call either closes the door or observes a closed door while executing 
        the $\lit{Doorway}$ procedure.
Therefore, the later call that we are linearizing as the winner has to observe a closed door on 
        line~\ref{line:doorcheck} and cannot avoid returning $\id{LOSE}$ on line~\ref{line:doorlose}.
Hence, this invocation can never return $\id{WIN}$, and since we are linearizing it as winner,
        it should be the case that it does not return and no other processor returns $\id{WIN}$.
We picked this invocation to have the earliest starting point, so every other $\lit{LeaderElect}$ 
        invocation that does not return must start after $P$.
Let us now consider an extension of the current execution where the processors executing these 
        invocations are continuosly scheduled to take steps and all messages are delivered.
According to the above argument, since all invocations start after $P$, these processors must 
        observe a closed door on line~\ref{line:doorcheck} and return $\id{LOSE}$ after only 
        finitely many steps.
We have hence constructed a valid execution where all processors that execute $\lit{LeaderElect}$ return $\id{LOSE}$.
This contradiction with~\lemmaref{lem:leastone} completes the proof.        
\end{proof}
\noindent We need one final claim before proving the main theorem.
\begin{claim}
\label{clm:constpart}
The maximum expected number of participants decreases 
        at least by some fixed constant fraction in every two rounds.
\end{claim}
\begin{proof}
This obviously holds for a single participant, because it will return $\id{WIN}$ in the next round
        and the number of participants after that will be zero.

We know that for $k$ participants in some round, by~\lemmaref{lem:loglow} and~\lemmaref{lem:loghigh},
        the maximum expected number of participants in the next round is $O(\log^2{k} + 1)$. 
This implies that for a large enough constant $D$, there is constant $c_1 < 1$ such that for 
        $k > D$ the maximum expected number of participants in the next round, and thus in all
        rounds thereafter, is at most $c_1k$.
If $k \leq D$, then the first processor that finishes executing line~\ref{line:commitpr1} flips $1$ 
        with at least a constant probability.
In this case, all processors that flip $0$ will die, and the expected number of the remaining
        processors that flip $0$ is at least $\frac{k-1}{2} \leq \frac{k}{4}$ for $k \geq 2$.
This is because the expected number of remaining processors that flip $1$ is at most $\frac{k-1}{2}$, 
        as each of them observes at least the first processor and itself, hence has no more than 
        $1/2$ probability of flipping $1$.
Thus, if $k \leq D$, with a constant probability, a constant fraction of participants dies, meaning that there is a 
        constant $c_2 < 1$ such that the maximum expected number of participants is at most $c_2k$.
Setting $c = \max(c_1, c_2) < 1$ we obtain that the maximum expected number of 
        participants in every two rounds always decreases by at least a constant fraction to $ck$.
\end{proof}
\begin{theorem}
\label{thm:maintas}
Our leader election algorithm is linearizable.
If there are at most $\lceil n/2 \rceil -1$ processor faults, all non-faulty processors 
        terminate with probability $1$.
For $k$ participants, it has time complexity $O(\log^{\ast}{k})$ and message complexity $O(kn)$.
\end{theorem}
\begin{proof}
We have shown linearizability in~\lemmaref{lem:elinear}.
        
All $k \geq 1$ processors participate in the first round.
The maximum expected number of processors that participate in round $3$ is 
        clearly no more than the maximum expected number of survivors of the first round, 
        which by~\lemmaref{lem:loglow} and~\lemmaref{lem:loghigh} for $k>1$ can be written
        as $C(\log^2{k} + 2\log{k})$ for some constant $C$.
If $k=1$, then this lone processor will observe all other processors in round $0$, leading to $R = 0$ and as current round 
 is $r = 2$ it will return $\id{WIN}$ in the second round.
Hence, for $k=1$, there will be zero participants in the third round.
Thus, for any $k$, the maximum expected number of participants in round $3$ is at most 
        $f(k) = C(\log^2{k} + 2\log{k})$.

Let us say the adversary can achieve a probability distribution for round $3$ such that
        there are $K_i$ participants with probability $p_i$.
We have shown above that 
\begin{equation}
\label{eq:1rexp}
\sum_i p_i K_i \leq f(k)
\end{equation}
\noindent Now, using the same argument as above, we can bound the maximum expected number 
        of participants in round $5$ to be at most $\sum p_i f(K_i)$.
Function $f$ is concave for non-negative arguments, and for arguments larger than a constant
        it is also monotonically increasing.
This implies that either $\sum p_i K_i$, the expected number of participants in round $3$, 
        is constant, or
\begin{equation}
\sum p_i f(K_i) \leq f( \sum p_i K_i) \leq f(f(k))
\end{equation}
where the first part is Jensen's inequality and the second follows from~\equationref{eq:1rexp} 
        and the monotonicity property. 
Similarly, we get that unless the maximum expected number of participants in round $5$ is less than 
        a constant, the maximum number of participants in round $7$ is at most $f(f(f(k)))$, 
        and so on.
Since $f(f(k)) \geq \log{k}$ for all $k$ larger than some constant, if we denote by $S_0$
        the number of participants in round $1 + 2\log^{\ast}{k}$, maximum $\mathbb{E}[S_0]$
        that the adversary can achieve must also be constant.
These $S_0$ participants execute the same algorithm, with $S_1$ of them participating 
        in the next round, etc.
        
Let $R$ be the number of remaining rounds. 
Expectation of $R$ can be written as
\begin{equation}
\mathbb{E}[R] = \sum_{i=1}^{\infty} \Pr[R \geq i] 
        = \sum_{i=1}^{\infty} \Pr[S_i \geq 1] 
        \leq \sum_{i=1}^{\infty} \mathbb{E}[S_i]
\end{equation}
where the equality is by the definition of rounds and then we apply Markov's inequality 
        to get to expectations.
Finally, by~\claimref{clm:constpart} we get that $\mathbb{E}[R] = O(\mathbb{E}[S_0]) = O(1)$
        and thus the maximum total number of rounds any processor participates in is 
        $O(\log^{\ast}{k})$, and processors perform only fixed, constantly many \lit{communicate} 
        calls per round.
Time complexity follows from~\claimref{clm:comtime}.
        
To bound the maximum expected number of messages, let $Q_r$ be the number of participants 
        in round $r$, counting from the very first round.
Since each processor sends $O(n)$ messages per round, the maximum expected number of messages is 
        $\sum_{r=1}^{\infty} \mathbb{E}[O(nQ_r)] = n\cdot \mathbb{E}[O(Q_1)] = O(nk)$ 
        using~\claimref{clm:constpart}.

If there are at most $\lceil n/2 \rceil -1$ processor faults, all \lit{communicate} calls 
        return, and processors must enter larger rounds.
However, the probability that all processors terminate before reaching round $r$ is 
        $1 - \Pr[Q_r \geq 1] \geq 1 - \mathbb{E}[Q_r]$ which tends to $1$ as
        $r$ increases by~\claimref{clm:constpart}.
\end{proof}
\subsection{Renaming Analysis}
\label{app:renproofs}
\begin{lemma}
\label{lem:renameterminate}
No two processors return the same name from the \lit{getName} call and if there are at most $\lceil n/2 \rceil -1$ processor faults, 
 all non-faulty processors terminate with probability 1.
\end{lemma}
\begin{proof}
Assume that less than half of the processors are faulty.
Processors executing the \lit{getName} call the \lit{communicate} procedure 
 which always terminates under at most at most $\lceil n/2 \rceil -1$ processor faults.
All local computations steps are also always performed successfully by non-faulty processors. 

Finally, processors invoke our leader election algorithm from~\sectionref{sec:tas} for at most $n$ names, 
 at most once for each name (the first time they set $\id{Contended} \gets \lit{true}$,
 which prohibits contending in the future).
By~\theoremref{thm:maintas} all invocations of the leader election for a particular name by non-faulty processors 
 terminate with probability $1$, and using union bound for at most $n$ names, 
 the probability that all leader election calls by all non-faulty processors terminate tends to $1$.
Therefore, with probability $1$, non-faulty processors keep making progress, i.e. they keep contending for new names,
 and as there are $n$ names and $n$ processors that do not contend for the same name twice, 
 each non-faulty processor eventually wins a name and returns.

A processor that returns some name $u$ from a \lit{getName} call
        has to be the winner of our leader election protocol.
However, according to~\theoremref{thm:maintas}, 
        $\lit{LeaderElect}_u$ cannot have more than one winner.
\end{proof}
\begin{lemma}
\label{lem:renameorder}
In any execution, if a processor views $\id{Contended}[i] = \lit{true}$ 
        in some while loop iteration, 
        and in some subsequent iteration on line~\ref{line:pickspot} the same 
        processor views $\id{Contended}[j] = \lit{false}$, $i \prec j$ has to hold.
\end{lemma}
\begin{proof}
Clearly, $j \neq i$ because contention information never disappears
        from a processor's view.
In the earlier iteration, the processor propagates 
        $\id{Contended}[i]=\lit{true}$ to a quorum on line~\ref{line:prepropcontend} or~\ref{line:propcontend}.
During the later iteration, on line~\ref{line:colContended},
        the processor collects information and does not set 
        $\id{Contended}[j]$ to $\lit{true}$ before reaching line~\ref{line:pickspot}.
Thus, more than half of the processors view
        $\id{Contended}[j] = \lit{false}$ at some intermediate time point.
Therefore, a quorum of processors views $\id{Contended}[i]=\lit{true}$
        strictly earlier than $\id{Contended}[j]=\lit{true}$, 
        and by definition $i \prec j$.
\end{proof}
\begin{lemma}
\label{lem:msgrenameweak}
The message complexity of our renaming algorithm is 
        $O(n \cdot \mathbb{E}[\sum_{i=1}^n X_i])$. 
\end{lemma}
\begin{proof}
Let $L_j$ be the number of loop iterations executed by processor $j$.
Then $\sum_i X_i = \sum_j L_j$, because every iteration involves
        one processor contending at a single name spot, and 
        no processor contends for the same name twice.
Each iteration involves two \lit{communicate} calls with $O(n)$ total messages.
The total number of messages sent in the leader election protocols is 
        $O(\sum_i n \cdot X_i)$.
The message complexity is thus the expectation of:
\begin{equation*}
O\left(\sum_i n \cdot X_i \right) + \sum_j O(n) \cdot L_j
 = O\left(\sum_i n \cdot X_i\right)
\end{equation*}
\noindent as desired.
\end{proof}
\begin{lemma}
\label{lem:ncontend}
In any execution, at most $\frac{n}{2^{j-1}}$ processors ever contend for names 
        from groups $G_{j' \geq j}$.
\end{lemma}
\begin{proof}
If no name from $G_{j' \geq j}$ is ever contended, then the statement is 
        trivially true.
If some name from $G_{j' \geq j}$ was contended, then by our ordering so were all names from the former groups. 
Otherwise, an uncontended name from an earlier group must be sorted later and cannot belong to an earlier group.

There are $n-\frac{n}{2^{j-1}}$ names in earlier groups.
Since they all were contended, there are $n-\frac{n}{2^{j-1}}$
        processors that can be linearized to win the corresponding 
        leader election and the name.
Consider one such processor $p$ and the name $u$ from some earlier group
        $G_{j''' < j}$, that $p$ is bound to win.
Processor $p$ does not contend for names after $u$, and it also 
        never contends for a name from $G_{j' \geq j}$ before contending 
        for $u$, because that contradicts~\lemmaref{lem:renameorder}.
Thus, none of the $n-\frac{n}{2^{j-1}}$ processors ever contend for a
        name from $G_{j' \geq j}$, out of $n$ processors in total, 
        completing the argument. 
\end{proof}
\begin{lemma}
\label{lem:rentype1}
$\mathbb{E}[\sum_{i=1}^n X_i(\id{clean})] = O(n).$
\end{lemma}
\begin{proof}
Let us equivalently prove that $\mathbb{E}[X_i(\id{clean})] = O(1)$ for 
        any name $u_i$ in some group $G_j$, where
        $X_i(\id{clean})$ is defined as the number of $\id{clean}(j)$ iterations, 
        in which processors contend for a name $u_i \in G_j$.

By definition, in all $\id{clean}(j)$ iterations a processor observes
        all names in $G_{j''>j}$ as uncontended on line~\ref{line:pickspot}.
Therefore, each time, independent of other iterations, the probability of picking spot $i$ and 
        contending for the name $u_i$ is at most $\frac{2^{j}}{n}$.
Thus, if there are exactly $\beta n$ of $\id{clean}(j)$ iterations, 
        $X_i(\id{clean}) \leq \mathrm{B}(\beta n, \frac{2^{j}}{n})$, thus
\begin{equation}
\mathbb{E}[X_i(\id{clean}) \mid u_i \in G_j, \beta n \textit{ iterations}] 
\leq 2^{j} \beta \label{lem:alphan}
\end{equation}
\noindent for $\beta n$ clean iterations that started in phase $j$. 
The probability that there are exactly $\beta n$ of $\id{clean}(j)$ iterations is trivially 
        upper-bounded by the probability that there are 
        at least $\beta n$ $\id{clean}(j)$ iterations, 
        which by~\corollaryref{cor:expdec} is at most 
        $e^{-\frac{\beta n}{32}}$ for $\beta \geq \frac{1}{2^{j-6}}$.
Therefore:
\begin{equation}
\mathbb{E}[X_i(\id{clean}) \mid u_i \in G_j] 
\leq    \frac{2^{j}}{2^{j-6}} + 
        \sum_{l = \lceil \frac{n}{2^{j-6}} \rceil}^{\infty} 
                e^{-\frac{l}{32}} \cdot \frac{2^{j}l}{n}
\end{equation}
\noindent which, after some calculation, is $O(1)$, completing the proof.
\end{proof}
\begin{claim}
\label{clm:notwice}
  In any execution, for each $j$, any processor participates in at most one $\id{dirty}(j)$
        and at most one $\id{cross}(j)$ iteration.
\end{claim}
\begin{proof}
  The first time processor $p$ participates in a $\id{dirty}(j)$ iteration, 
        by definition, it views $\id{Contended}[i] = \lit{true}$ for some $u_i \in G_{j''>j}$.
  Therefore, $p$ also propagates $\id{Contended}[i] = \lit{true}$ 
        on line~\ref{line:prepropcontend} in the same iteration.
  When $p$ starts a subsequent iteration, a quorum of processors know about 
        $u_i \in G_{j''>j}$ being contended.
  By the way names in $u$ are sorted, at that point more than half of the processors 
        must already know that each name in $G_j$ is contended, meaning that phase $j$ has ended.
  Therefore, no subsequent iteration of the processor can be of a $\id{dirty}(j)$ type.
        
  On the other hand, when a processor completes a $\id{cross}(j)$ iteration, it has propagated
        contention information for $u_i \in G_j$ to a quorum, meaning that because of the way
        names in $u$ are sorted, phase $j$ must have been started already, 
        and no operation that starts later can be $\id{cross}(j)$.
\end{proof}
\begin{lemma}
\label{lem:rentype23}
$\mathbb{E}[\sum_{i=1}^n X_i(\id{dirty})] = O(n)$ and 
        $\mathbb{E}[\sum_{i=1}^n X_i(\id{cross})] = O(n).$
\end{lemma}
\begin{proof}
  Recall that $X_i(\id{dirty})$ is the number of processors 
        that ever contend for $u_i \in G_j$ in a $\id{dirty}(j)$ iteration.
  Let us equivalently fix $j$ and prove that 
        $\mathbb{E}[\sum_{u_i \in G_j} X_i(\id{dirty})] = O(\frac{n}{2^{j-1}})$, 
        which implies the desired statement by linearity of expectation and telescoping.

  We sum up quantities $X_i(\id{dirty})$ for the names in $j$-th group, 
        but the adversary controls precisely which names belong to group $G_j$.
  We will therefore consider all names $u_i \in G_{j' \geq j}$ and sum up quantities $X_{i,j}$:
        the number of processors that contend for a name $u_i$ in a $\id{dirty}(j)$ iteration.
  All $\id{dirty}(j)$ iterations by definition start in phase $j$, and by~\lemmaref{lem:ncontend} 
        there can be at most $\frac{n}{2^{j-1}}$ different processors executing them.
  Moreover, by~\claimref{clm:notwice} each of these processors can participate in 
        at most one $\id{dirty}(j)$ iteration, implying $\sum_{u_i} X_{i,j} \leq \frac{n}{2^{j-1}}$.
  Thus $\mathbb{E}[\sum_{i=1}^n X_i(\id{dirty})] = 
        \sum_{j=1}^{\log{n}} \sum_{u_i \in G_{j'\geq j}} X_{i,j} = O(n)$ as desired.
        
  The proof for $X_i(\id{cross})$ is analogous because at most $\frac{n}{2^{j-1}}$ 
        different processors contend for names $u_i \in G_{j' \geq j}$ by~\lemmaref{lem:ncontend}, 
        each participating in at most one $\id{cross}(j)$ iteration by~\claimref{clm:notwice}.
\end{proof}
\begin{theorem}
\label{thm:rounds}
The time complexity of the the renaming algorithm is $O( \log^2 n )$.
\end{theorem}
\begin{proof}
We will prove that the maximum expected number of \lit{communicate} calls by any processor
        that the adaptive adversary can achieve is $O( \log^2 n )$, which implies the result by~\claimref{clm:comtime}.
	
In the following, we fix an arbitrary processor $p$, 
        and upper bound the number of loop iterations it performs during the execution. 
Let $M_i$ be the set of free slots that $p$ sees when performing 
        its random choice in the $i$th iteration of the loop, and let $m_i = |M_i|$. 
By construction, notice that there can be at most $m_i$ processors that compete with for slots 
        in $M_i$ for the rest of the execution. 

Assuming that $p$ does not complete in iteration $i$, let $Y_i \subseteq M_i$ be the set of 
        \emph{new} slots that $p$ finds out have become contended at the beginning 
        of iteration $i + 1$, and let $y_i = |Y_i|$. 
We define an iteration as being \emph{low-information} if $y_i / m_i < 1 / \log m_i$.  
Notice that, in an iteration that is high-information, the processor might collide, 
        but at least reduces its random range for choices by a $1  / \log m_i$ factor.  

Let us now focus on low-information iterations, and in particular let $i$ be such an iteration. 
Notice that we can model the interaction between the algorithm and the adversary 
        in iteration $i$ as follows. 
Processor $p$ first makes a random choice $r$ from $m_i$ slots it sees as available. 
By the principle of deferred decisions, we can assume that, at this point, 
        the adversary schedules all other $m_i - 1$ processors to make their choices in this round, 
        from slots in $M_i$, with the goal of  causing a collision with $p$'s choice. 
        (The adversary has no interest in showing slots outside $M_i$ to processors.)
Notice that, in fact, the adversary may choose to schedule certain processors 
        multiple times in order to obtain collisions. 
However, by construction, each re-scheduled processor announces its choice in the iteration 
        to a quorum, and this choice will become known to $p$ in the next iteration. 
Therefore, re-scheduled processors should not announce more than $m_i / \log m_i$ distinct slots. 
Intuitively, the number of re-schedulings for the adversary can be upper bounded by 
the number of balls falling into the $m_i / \log m_i$ most loaded bins 
in an $m_i - 1$ balls into $m_i$ bins scenario. 
A simple balls-and-bins argument yields that, in any case, the adversary cannot perform 
        more than $m_i$ re-schedules without having to announce 
        $m_i / \log m_i$ new slots, with high probability in $m_i$. 

Recall that the goal of the adversary is to cause a collision with $p$'s random choice $r$. 
We can reduce this to a balls-into-bins game in which the adversary throws $m_i - 1$ 
        initial balls and an extra $m_i$ balls (from the re-scheduling) 
into a set of $m_i ( 1 - 1 / \log m_i)$ bins, with the goal of hitting a specific bin, 
        corresponding to $r$. 
        (The extra $( 1 - 1 / \log m_i)$ factor comes from the fact that certain processors 
        (or balls) may already observe the slots removed in this iteration.) 
The probability that a fixed bin gets hit is at most 
        $$ \left(1 - \frac{1}{m_i ( 1 - 1 / \log m_i )}\right)^{2m_i} \leq (1 / e)^3.$$ 

Therefore, processor $p$ terminates in each low-information iteration with constant probability. 
Putting it all together, we obtain that, for $c \geq 4$ constant, after $c \log^2 n / \log \log n$ iterations, 
        any processor $p$ will terminate with high probability, 
        either because $m_i = 1$ or because one of its probes 
        was successful in a low-information phase.

In each loop iteration, a processor performs a fixed constant additional number of \lit{communicate} calls on top
 of the \lit{communicate} calls performed while executing the leader election algorithm for the name picked in that iteration.
By~\theoremref{thm:maintas}, the maximum expected number of \lit{communicate} calls in each leader election 
 is $O(\log^{\ast} n)$, and by linearity of expectation, total maximum number of \lit{communicate} calls by any processor
 is at most $O(\frac{\log^2 n \log^{\ast} n}{\log \log n}) = O(\log^2 n)$.
\end{proof}

\section{Message Complexity Lower Bounds}
\label{app:lbounds}
In this section, we prove that our algorithms are message-optimal by showing a lower bound of expected $\Omega( n^2 )$ messages on any algorithm implementing leader election or renaming in an asynchronous message-passing system where $t < n / 2$ processors may fail by crashing. In fact, we prove such a lower bound for any object with \emph{strongly non-commutative} methods~\cite{LawsOfOrder}. 

\begin{definition}
        Given an object $O$, a method $M$ of this object is \emph{strongly non-commutative} if there exists some state $S$ of $O$ for which an instance $m_1$ of $M$ executed sequentially by processor $p$ changes the result of an instance $m_2$ of $M$ executed by processor $q \neq p$, and vice-versa, i.e. $m_2$ changes the result of $m_1$ from state $S$.
\end{definition}

\noindent We now give a message complexity lower bound for objects with non-commutative operations. 

\begin{theorem}
\label{thm:lb}
Any implementation of an object $O$ with a strongly non-commutative operation $M$ by $k \leq n$ 
        processors guaranteeing termination with at least constant probability $\alpha > 0$ 
        in an asynchronous message-passing system where $t < n / 2$ processors may 
        fail by crashing must have worst-case expected message complexity $\Omega( \alpha k n )$. 
\end{theorem}
\begin{proof}
Let $A$ be an algorithm implementing a shared object $O$ with a strongly non-commutative method $M$, 
        in asynchronous message-passing with $t < n / 2$, guaranteeing termination with probability $\alpha$. 
We define an adversarial strategy for which we will argue that all the resulting terminating executions 
        (regardless of their probability) must cause $\Omega( kn )$ messages to be sent. 
This clearly implies our claim. 
The strategy proceeds as follows.
        
Assume that each processor is executing an instance of $M$. 
The adversary picks a subset $S$ of $k / 4$ participants, and places them in a ``bubble:'' 
        for each such processor $q$, the adversary suspends all its incoming and outgoing messages 
        in a buffer, until there are at least $n / 4$ such messages in the buffer. 
At this point, the processor is freed from the bubble, and is allowed to take steps synchronously, 
        together with other processors. 
Processors outside $S$ execute in lock-step, 
        and their messages outside the bubble are delivered in a timely fashion. 
        
Note that this strategy induces a family of executions $\mathcal{E}$, 
        each of which is defined by the set of coin flips made by the processors. 
We can assume that there exists a time $\tau$ after which in all executions in $\mathcal{E}$ with non-zero probability
 no processors send any more messages.
Otherwise, the adversary can always wait for another message that must be sent, 
 then for the next message, and so on, until $\Omega( k n )$ messages.

Let us prove that in each execution $E \in \mathcal{E}$ every processor in the bubble must eventually leave the bubble before returning, 
 which implies $\Omega( k n )$ messages in executions in which all processors return.
Towards this goal, we first show that a processor cannot return while still in the bubble.
Then we prove that all processors in the bubble are forced to either return while still in the bubble (which cannot happen)
 or leave the bubble, completing the proof.

For the first part, assume for the sake of contradiction that there exists an execution 
        $E \in \mathcal{E}$ and a processor $p \in S$ that decides in $E$ 
        while still being in the bubble. 
Practically, this implies that $p$ has returned from its method invocation without receiving any 
        messages, and without any of its messages being received. 
To obtain a contradiction, we build two alternate executions $E'$ and $E''$, both of which are 
        indistinguishable to $p$, but in which $p$ must return different outputs. 
        
In execution $E'$, we run all processors outside the bubble until one of them returns--this must eventually occur with constant probability, 
 since this execution is indistinguishable to these processors from an execution in which all (at most $k/4 < n/2$) 
 processors in the bubble 
 are initially crashed.
We suspend messages sent to the processors inside the bubble.
We then run processor $p$, which flips the same coins as in $E$ (the execution exists as this happens with probability $>0$), 
 observes the same emptiness and therefore eventually returns with constant probability, without having received any messages.
We deliver $p$'s messages and suspended messages as soon as $p$ decides.
         
In execution $E''$, we first run $p$ in isolation, suspending its messages.
With probability $>0$, $p$ flips the same coins as in $E$, and must eventually decide with constant probability without 
        having received any messages.
We then run all processors outside the bubble in lock-step. 
One of these processors must eventually return with constant probability, since to these processors, 
 the execution is indistinguishable from an execution in which $p$ (and other processors in the bubble) 
 has crashed initially. 
We deliver $p$'s messages after this decision. 
Since both $E'$ and $E''$ are indistinguishable to $p$, 
        it has to return the same value in both executions with constant probability. 
However, this cannot be the case because instances of method $M$ are strongly non-commutative, 
        the two returning instances are not concurrent, and occur in opposite orders in the two executions. This correctness requirement is enforced \emph{deterministically}.
Therefore, $p$ must return distinct values in  executions $E'$ and $E''$, which is a contradiction. 
Hence, $p$ cannot return in $E$.

To complete the argument, we prove that $p$ has to eventually return or leave the bubble, with probability $\geq \alpha$.
We cannot directly require this of the execution prefix $E$ since not all messages 
        by correct processors have been delivered in this prefix. 
For this, we consider time $\tau$, at which we crash all recipients of messages by $p$, 
        and all processors that sent messages to $p$ in $E$.
By the definition of the bubble, the number of crashes we need to expend is $< n / 4$. 
Therefore, by definition of $\tau$, there exists a valid execution, in which no more messages will be sent and 
 $p$ must eventually decide with probability $\geq \alpha$. 
From $p$'s prospective, the current execution in the bubble can be this execution, and if the adversary keeps $p$
 in the bubble for long enough, it has to decide with probability $\geq \alpha$.
However, from the previous argument, we know that $p$ cannot decide while in the bubble,
        therefore $p$ has to eventually leave the bubble in order to be able to decide and return.
        
This shows that a specific processor $p$ must eventually leave the bubble. 
The final difficulty is in showing that we can apply the same argument to \emph{all} processors 
        in the bubble at the same time without exceeding the failure budget. 
Notice however that we could apply the following strategy: for each processor $q$ in the bubble, 
        we could fail all senders and recipients of $q$ ($< n / 4$), 
        and also all other processors in the bubble ($< n / 4$) at time $\tau$. 
This can be applied without exceeding the failure budget. 
Since any processor $q$ could be the sole survivor from the bubble to which we have applied the buffering strategy,
 and since $q$ does not see a difference from an execution in which it has to return, 
 analogously to the previous case, we obtain that each $q$ in the bubble has to eventually leave the bubble with probability $\geq \alpha$.
        
Therefore, we obtain that at least $\alpha k n / 16$ messages have to be exchanged during the execution,
        which implies the claim. 
\end{proof}

It is easy to check that the \emph{elect} procedure of a leader election algorithm and the \emph{rename} procedure of a strong renaming algorithm 
are both non-commutative. (In the case of renaming, consider $n + 1$ distinct processors executing the rename procedure. 
By the pigeonhole principle, there exists some non-zero probability that two processors choose the same name in solo executions. Therefore, these two operations do not commute, and therefore the \emph{rename} procedure is strongly non-commutative.)
We therefore obtain the following corollary.

\begin{corollary}
\label{cor:lowerbound}
	Any implementation of leader election or renaming by $k \leq n$ 
        processors which ensures termination with probability at least $\alpha > 0$ in an asynchronous message-passing system where $t < n / 2$ processors may 
        fail by crashing must have worst-case expected message complexity $\Omega( \alpha k n )$. 
\end{corollary}

\end{document}